\documentclass[a4paper,twoside,12pt]{article}
\usepackage[T1, T2A]{fontenc}
\usepackage[english]{babel}
\usepackage{amsthm}
\usepackage{amsmath}
\usepackage{graphicx}
\usepackage{wrapfig}
\usepackage{caption}
\usepackage{bm}
\usepackage{subcaption}
\usepackage{amssymb}
\usepackage{thmtools}

\newtheorem{theorem}{Theorem}
\newtheorem{problem}{Problem}
\newtheorem{proposition}{Proposition}
\newtheorem{corollary}{Corollary}
\newtheorem{lemma}{Lemma}
\theoremstyle{remark} \newtheorem{zabelejka}{Remark}

\usepackage{mathtext}
\usepackage{mathtools}
\usepackage{breqn}
\usepackage{authblk}
\title{A Calibration Algorithm for Microelectromechanical Systems Accelerometers in Inertial Navigation Sensors}
\author[1]{Svetoslav Nakov}
\author[2]{Tihomir Ivanov}
\affil[1]{Institute of Mechanics and Biomechanics, Bulgarian Academy of Sciences,\textit{ sisqo\textunderscore nakov@yahoo.com}}
\affil[2]{Institute of Mathematics and Informatics, Bulgarian Academy of Sciences, \textit{tihomir\textunderscore ivanov@ymail.com}}
\begin{document}
\maketitle
\begin{abstract}
In the present work we develop an algorithm for calibrating MEMS sensors, which accounts for the nonorthogonality of the accelerometers' axis, as well as for the constant bias and scaling errors. We derive an explicit formula for computing the calibrated acceleration, given data from the sensors. We also study the error, that is caused by the nonorthogonality of the axis. 
\end{abstract}
\section{Introduction}
Nowadays, systems, using microelectromechanical systems (MEMS) sensors (e.g. accelerometers, gyroscopes, etc.) find more and more applications. For instance, they are used for video stabilization, mobile applications and many other things. Among the reasons for their wide use is the fact that they are very cheap and have miniature size (usually less than 10x10x5 mm) and consume little power. However, there are a lot of unexplored fields (like navigation) where they could be used once the necessary mathematical knowledge is developed. For further information, see for example \cite{Woodman}. 

In the present work  we address MEMS accelerometers in particular. There are many sources of error that affect their behaviour \cite{Woodman}. We develop an algorithm for calibrating the device, i.e. compensating for the constant bias error (the offset of its output signal, when the device does not undergo any acceleration), scaling errors (the deviation of the accelerometer's scale from the unit we want to use, like $m/s^2$) and nonorthogonality of the axis.

We have a system, consisting of three accelerometers, that measure the linear acceleration in three ``almost'' orthogonal directions. We shall assume that the angle between each two of those directions is $90^\circ\pm2\%$ \cite{Specification}. 
The relationship between accelerometer's outputs and the measured acceleration can be well modeled as a linear function \cite{Chatfield}. In Section 2 we present the calibration algorithm. Similar approach to the problem is proposed, e.g., in \cite{Forsberg/Grip}, \cite{Grip/Sabourova} (see also the references within those papers). There a six-parameter and a nine-parameter linear models are developed for calibrating the device. Both of those approaches use an orthogonal coordinate system, with respect to which the corresponding analysis is carried. We study the problem by working in the coordinate system defined by the directions in which the accelerometers measure linear acceleration. We find this to be the more natural way and, thus, simplifying the analysis. Also, an usual thing to do is to approximate the sine and the cosine functions of small angles with the respective angles (which leads to the linearity of the aforementioned models). We do not use this approximation in order to obtain maximal accuracy in the model. 

In Section 3 we study computationally the error that can be caused by not considering the non-orthogonality of the axis. We show that even small deviations in the directions of the accelerometers can lead to catastrophic error in the estimated position of the body within a few seconds of integrating.

\section{Mathematical Model}
Let $\hat{a_x}$, $\hat{a_y}$ and $\hat{a_z}$ be the real values
that the accelerometers show and $a_x$, $a_y$ and $a_z$ be the corresponding calibrated values of the acceleration in the $x$, $y$ and $z$ directions, respectively, with $\vec{e_1}$, $\vec{e_2}$ and $\vec{e_3}$ being unit vectors in those directions. We use the following model to calibrate the accelerometers \cite{Chatfield}:
\begin{equation}
\begin{aligned}
a_x&=(\hat{a_x}-s_1)/b_1\\
a_y&=(\hat{a_y}-s_2)/b_2\\
a_z&=(\hat{a_z}-s_3)/b_3
\end{aligned},
\label{theModel}
\end{equation}
where $s_1$, $s_2$, $s_3$ (which we call shifts) and $b_1$, $b_2$, $b_3$ (which we call scaling 
coefficients) are yet to be determined.
In matrix form (\ref{theModel}) can be written as
\[
\underbrace{
\left[
\begin{array}{c}
a_x\\
a_y\\
a_z
\end{array}
\right]}_{\bm{a}}=
\underbrace{
\left[
\begin{array}{ccc}
1/b_1&0&0\\
0&1/b_2&0\\
0&0&1/b_3
\end{array}
\right]}_{T}
\underbrace{
\left[
\begin{array}{c}
\widehat{a_x}\\
\widehat{a_y}\\
\widehat{a_z}
\end{array}
\right]}_{\bm{\widehat{a}}}-
\underbrace{\left[
\begin{array}{c}
s_1/b_1\\
s_2/b_2\\
s_3/b_3
\end{array}
\right]}_{\bm{s}}.
\]
Using the introduced notation, we obtain the equation
\begin{equation}
\bm{a}=T\bm{\widehat{a}}-\bm{s}
\label{theModelMatrix}
\end{equation}
\begin{proposition}
If we assume that the axis of the accelerometers are orthogonal, then the acceleration is
\[
a_{orth}:=\sqrt{a_x^2+a_y^2+a_z^2}
\]
\end{proposition}
\begin{proof}
The proof is obvious, using the fact that the sum of the squares of the directional cosines is 1.
\end{proof}
Since the axis are not orthogonal, let us denote the angles between the axis as follows (see Fig.\ref{fig}):
\begin{equation}
	\phi:=\angle(Ox,Oy),\ \psi:=\angle(Ox,Oz),\ \theta:=\angle(Oy,Oz)
\label{angles}	
\end{equation}	
Let $\vec{a}$ be an arbitrary acceleration acting on the sensor. Let $\overline{a_x}$, $\overline{a_y}$ and
$\overline{a_z}$  be the affine projections of $\vec{a}$ onto the $x$, $y$ and $z$ axis, respectively.
It is important to note 
that the measured values\footnote{For the time being, let us think that the accelerometeres measure correctly, i.e. $a_x$, $a_y$ and $a_z$ are the measured values from the accelerometers.}  of the acceleration in the $x$, $y$ and $z$ directions are not the affine projections of $\vec{a}$ onto the axis.
Let us denote (see Fig.\ref{fig})
\[\alpha:=\angle(\vec{a},\vec{e_1}), \beta:=\angle(\vec{a},\vec{e_2}),\gamma:=\angle(\vec{a},\vec{e_3})\]
\begin{figure}[h!]
\begin{center}
\includegraphics[width=0.75\textwidth]{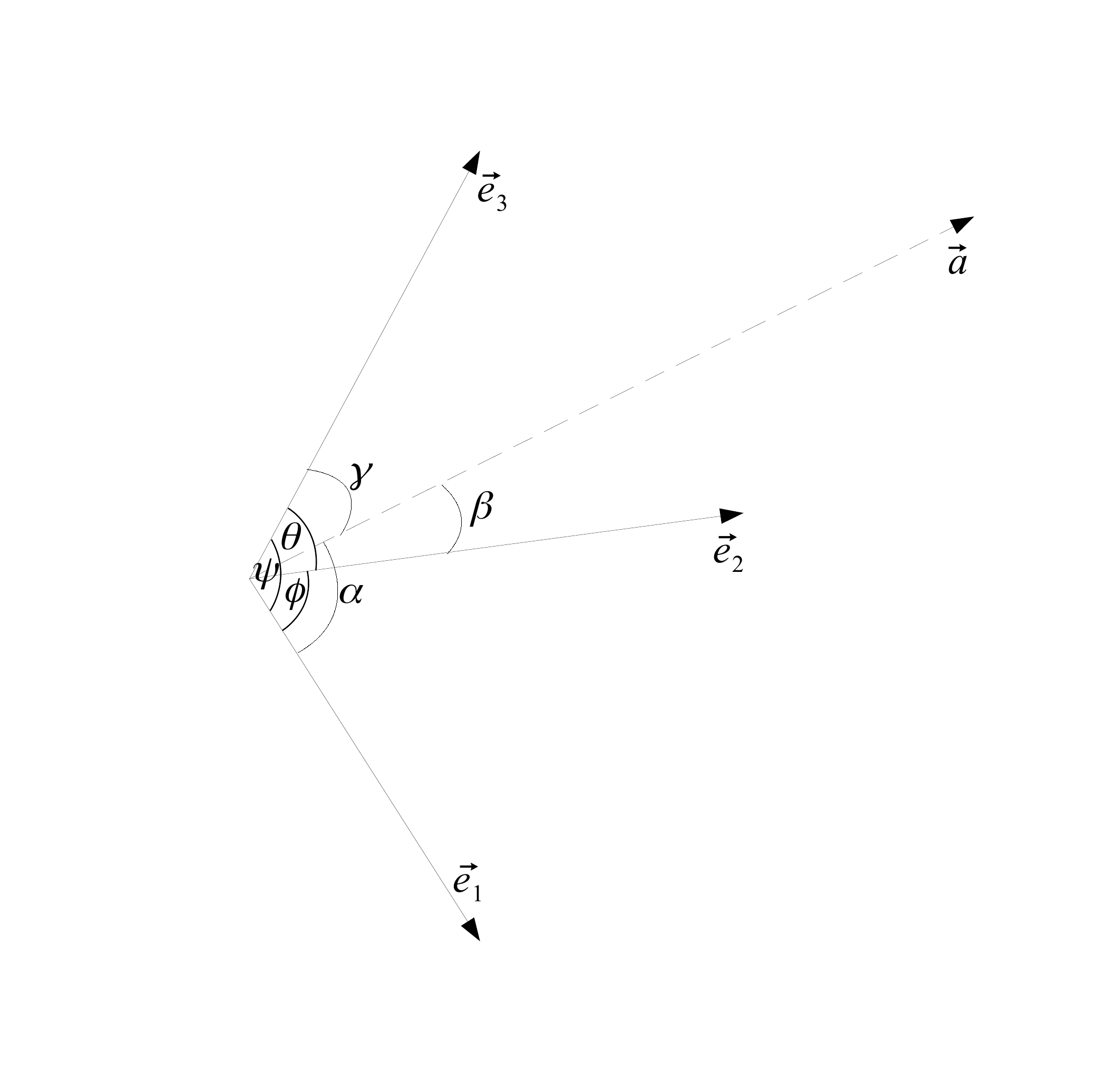}
\end{center}
\caption{Acceleration vector in the affine coordinate system}
\label{fig}
\end{figure}
It can be easily seen that the following Lemma holds true:
\begin{lemma}
The equations
\begin{equation}
\begin{aligned}
a_x&=\|\vec{a}\|\cos\alpha\\
a_y&=\|\vec{a}\|\cos\beta\\
a_z&=\|\vec{a}\|\cos\gamma
\end{aligned}
\label{eq:1}
\end{equation}
are valid.
\end{lemma}
We want to express the affine projections $\overline{a_x}$, $\overline{a_y}$ and
$\overline{a_z}$ of $\vec{a}$ onto the coordinate axis with $a_x$, $a_y$ and $a_z$. 
\begin{lemma}
\label{lemma2}
For the affine projections of the acceleration $\overline{a_x}$, $\overline{a_y}$ and
$\overline{a_z}$ the following equalities hold true:
\begin{equation}
\begin{aligned}
\overline{a_x}&=\frac{-\sin^2\theta a_x+(\cos\phi-\cos\theta\cos\psi)a_y
+(\cos\psi-\cos\phi\cos\theta)a_z}{-1+\cos^2\phi+\cos^2\psi+\cos^2\theta-
2\cos\phi\cos\psi\cos\theta}\\
\overline{a_y}&= \frac{(\cos\phi-\cos\theta\cos\psi)a_x-\sin^2\psi a_y
+(\cos\theta-\cos\phi\cos\psi)a_z}{-1+\cos^2\phi+\cos^2\psi+\cos^2\theta-
2\cos\phi\cos\psi\cos\theta}\\
\overline{a_z}&= \frac{(\cos\psi-\cos\phi\cos\theta)a_x+(\cos\theta-\cos\phi\cos\psi) a_y
-\sin^2\phi a_z}{-1+\cos^2\phi+\cos^2\psi+\cos^2\theta-
2\cos\phi\cos\psi\cos\theta}
\end{aligned}
\label{affineByReal}
\end{equation}
\end{lemma}

\begin{proof}
We have 
\begin{equation}
\vec{a}=\overline{a_x}\vec{e_1}+\overline{a_y}\vec{e_2}+\overline{a_z}\vec{e_3},
\label{eq:2}
\end{equation}
where $\vec{e_1}$, $\vec{e_2}$ and $\vec{e_3}$ are the unit vectors in the $x$, $y$ and $z$ directions,
respectively. Consecutively, we take the scalar products of the both sides of (\ref{eq:2}) with $\vec{e_1}$, $\vec{e_2}$, $\vec{e_3}$ and obtain:
\begin{equation}
\begin{aligned}
(\vec{a},\vec{e_1})=\overline{a_x}+\overline{a_y}\cos\phi+\overline{a_z}\cos\psi\\
(\vec{a},\vec{e_2})=\overline{a_x}\cos\phi+\overline{a_y}+\overline{a_z}\cos\theta\\
(\vec{a},\vec{e_3})=\overline{a_x}\cos\psi+\overline{a_y}\cos\theta+\overline{a_z}
\end{aligned}
\label{eq:3}
\end{equation}
Taking into account (\ref{eq:1}) we obtain 
\begin{equation}
\left[
\begin{array}{ccc}
1&\cos\phi&\cos\psi\\
\cos\phi&1&\cos\theta\\
\cos\psi&\cos\theta&1
\end{array}
\right]
\left[
\begin{array}{c}
\overline{a_x}\\
\overline{a_y}\\
\overline{a_z}
\end{array}
\right]=
\left[
\begin{array}{c}
a_x\\
a_y\\
a_z
\end{array}
\right]
\label{eq:12}
\end{equation}
We solve the system (\ref{eq:12}) and obtain (\ref{affineByReal}).
\end{proof}
\begin{zabelejka}
Let us note that the system (\ref{eq:12}) has diagonally dominant matrix and, therefore, it has a unique solution.
\end{zabelejka}
\begin{zabelejka}
In matrix form, equations (\ref{affineByReal}) can be rewritten as
\[\bm{\overline{a}}=\overline{T}\bm{a},\]
where:
\[
\begin{aligned}
&\bm{\overline{a}}=\left[\begin{array}{c}
\overline{a_x}\\
\overline{a_y}\\
\overline{a_z}
 \end{array}\right],\ 
 \overline{T}=\frac{1}{den}\left[
\begin{array}{ccc}
-\sin^2\theta &\cos\phi-\cos\theta\cos\psi&\cos\psi-\cos\phi\cos\theta\\
\cos\phi-\cos\theta\cos\psi&-\sin^2\psi&\cos\theta-\cos\phi\cos\psi\\
\cos\psi-\cos\phi\cos\theta&\cos\theta-\cos\phi\cos\psi&-\sin^2\phi
\end{array}
\right],\\
&den={-1+\cos^2\phi+\cos^2\psi+\cos^2\theta-
2\cos\phi\cos\psi\cos\theta},\ 
\bm{a}=\left[\begin{array}{c}
a_x\\
a_y\\
a_z
\end{array}\right]
\end{aligned}
\]
\label{remark:1}
\end{zabelejka}
\begin{lemma}
For the acceleration $a_{nonorth}$ the following holds true
\begin{equation}
a_{nonorth}=\sqrt{\overline{a_x}^2+\overline{a_y}^2+\overline{a_z}^2+2\overline{a_x}\overline{a_y}\cos\phi+2\overline{a_x}\overline{a_z}\cos\psi+2\overline{a_y}\overline{a_z}\cos\theta}.
\label{a_nonorth}
\end{equation}
\end{lemma}
\begin{proof}
Taking $a=a_{nonorth}=(\overline{a_x},\overline{a_y},\overline{a_z})$ in (\ref{eq:2}) and taking the 
scalar squares of both sides, we obtain
\[
(a_{nonorth},a_{nonorth})=\overline{a_x}^2+\overline{a_y}^2+\overline{a_z}^2+2\overline{a_x}\overline{a_y}(\vec{e_1},\vec{e_2})+
2\overline{a_x}\overline{a_z}(\vec{e_1},\vec{e_3})+2\overline{a_y}\overline{a_z}(\vec{e_2},\vec{e_3}),
\]
where $\vec{e_1}$, $\vec{e_2}$, $\vec{e_3}$ are the unit vectors along the axis. Taking into account (\ref{angles}), (\ref{a_nonorth}) follows directly.
\end{proof}
We substitute (\ref{affineByReal}) into (\ref{a_nonorth}) and simplify to obtain the following:
\begin{proposition}
The equality
\begin{equation}
a_{nonorth}=\sqrt{\frac{numerator}{denominator}},
\label{aNonorth}
\end{equation}
where
\[
\begin{aligned}
	numerator=&2(-1+1\cos2\theta)a_x^2+2(-1+1\cos2\psi)a_y^2+2(-1+1\cos2\phi)a_z^2\\
	&+8(\cos\phi-\cos\psi\cos\theta)a_xa_y+8(\cos\theta-\cos\phi\cos\psi)a_ya_z\\
	&+8(\cos\psi-\cos\phi\cos\theta)a_xa_z;\\
denominator=&2+2\cos2\phi+2\cos2\psi+2\cos2\theta-8\cos\phi\cos\psi\cos\theta
\end{aligned}
\]
holds true.
\label{theProposition}
\end{proposition}

Now we are ready to explain the algorithm for estimating the shifts and the
scaling coefficients in (\ref{theModel}). When at rest the sensor should measure the gravitational
acceleration $g$. Therefore, from Proposition \ref{theProposition} it follows that 
\[g=a_{nonorth},\]
where $a_{nonorth}$ is defined by (\ref{aNonorth}).
We substitute (\ref{theModel}) in $a_{nonorth}$ and obtain $\overline{a_{nonorth}}(s_1,s_2,s_3,b_1,b_2,b_3,\phi,\psi,\theta)$.  We use a least-squares fit to
find the values for $s_1,s_2,s_3,b_1,b_2,b_3,\phi,\psi,\theta$.

\section{Expressing the Linear Acceleration in an Orthonormal System}
In this section we shall derive an explicit expression for the calibrated values of the acceleration in an orthonormal cooridnate system, given the data, measured by the sensors. 
\begin{proposition}
\label{prop:orthSyst}
Let $\vec{e_1}$, $\vec{e_2}$, $\vec{e_3}$ be the unit vectors in the directions, defined by the accelerometers' axis. Let us introduce an orthonormal coordinate system $O_{\vec{f_1},\vec{f_2},\vec{f_3}}$, where $\vec{f_1}\equiv \vec{e_1}$, $\vec{f_2}$ is in the $\vec{e_1}\vec{e_2}$ plane. Then the following equality is valid:
\begin{equation}
\left[
	\begin{array}{c}
		\vec{e_1}\\
		\vec{e_2}\\
		\vec{e_3}
	\end{array}
\right]=\overline{T}^{orth} \left[
	\begin{array}{c}
		\vec{f_1}\\
		\vec{f_2}\\
		\vec{f_3}
	\end{array}
\right],
\end{equation}
where
\begin{equation}
\overline{T}^{orth}=\left[\begin{array}{ccc}
		1&\cos\phi&\cos\psi\\ && \\
		0&\sin\phi&\dfrac{\cos\theta-\cos\phi\cos\psi}{\sin\phi}\\ && \\
		0&0&\dfrac{\sqrt{1-\cos^2\phi-\cos^2\psi-\cos^2\theta+2\cos\phi\cos\psi\cos\theta}}{\sin\phi}
	\end{array}\right]
\label{eq:125}
\end{equation}	
\end{proposition}
\begin{proof}
Let us denote (see Fig. \ref{fig:orth-aff1})
\[
	\xi:=\angle(\vec{e_3},\vec{f_2}), \ \eta:=\angle(\vec{e_3}, \vec{f_3}).
\]
\begin{figure}[!h]
\centering
\includegraphics[width=0.7\textwidth]{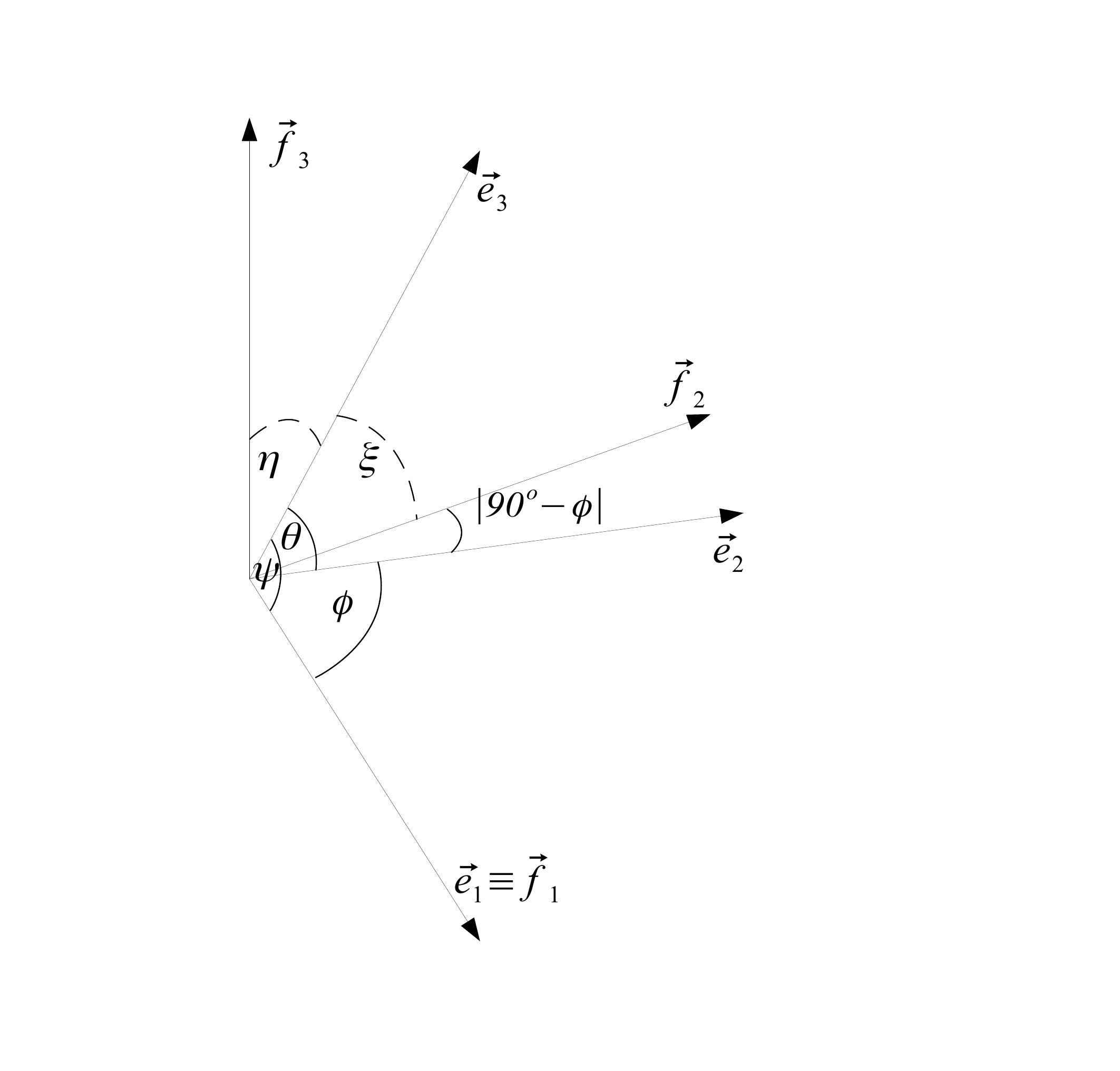}
\caption{Orthonormal coordinte system}
\label{fig:orth-aff1}
\end{figure}

Since $\vec{f_2}$ lies in the $\vec{e_1}\vec{e_2}$ plane and $\vec{f_1}\equiv\vec{e_1}$ we have the relationship
\begin{equation}
\vec{e_2}=b_1 \vec{f_1} + b_2 \vec{f_2},
\label{eq:123}
\end{equation}
where $b_1$ and $b_2$ are real numbers. By taking the inner product of (\ref{eq:123}) successively with $\vec{f_1}$ and $\vec{f_2}$ we obtain
\[
\begin{aligned}
b_1=\cos\phi, \ 
b_2=\cos |90^\circ-\phi|=\sin\phi
\end{aligned}
\]
Therefore, we have 
\[\vec{e_2}=\cos\phi\vec{f_1}+\sin\phi\vec{f_2}\]
and, thus,
\begin{equation}\vec{f_2}=-\frac{\cos\phi}{\sin\phi}\vec{e_1}+\frac{1}{\sin\phi}\vec{e_2}.
\label{eq:124}
\end{equation}
Analogously, we obtain
\begin{equation*}
\vec{e_3}=\cos\psi\vec{f_1}+\cos\xi\vec{f_2}+\cos\eta\vec{f_3},
\end{equation*}
By taking into account (\ref{eq:124}), we obtain
\[
\begin{aligned}
\cos\xi&=\left<\vec{e_3},\vec{f_2}\right>=\left<\vec{e_3},-\frac{\cos\phi}{\sin\phi}\vec{e_1}+\frac{1}{\sin\phi}\vec{e_2}\right>\\
&=-\frac{\cos\phi}{\sin\phi}\cos\psi+\frac{\cos\theta}{\sin\phi}=\frac{\cos\theta-\cos\phi\cos\psi}{\sin\phi}.
\end{aligned}
\]
Furthermore, since $\vec{e_3}$ is a vector in the orthonormal coordinate system $O_{\vec{f_1}\vec{f_2}\vec{f_3}}$ and the angles with the axis being $\psi$, $\xi$ and $\eta$, respectively the following holds true:
\[
\begin{aligned} 
&\cos^2\eta+\cos^2\xi+\cos^2\psi=1\\
&\cos\eta=\sqrt{1-\cos^2\xi+\cos^2\psi}\\
&=\dfrac{\sqrt{1-\cos^2\phi-\cos^2\psi-\cos^2\theta+2\cos\phi\cos\psi\cos\theta}}{\sin\phi}
\end{aligned}
\]  
Let us note, that we have used the fact that $\eta$ is small angle and, therefore, $\cos\eta>0$.
\end{proof}
\begin{corollary}
Let us have an arbitrary acceleration $\vec{a}$ with coordinates in the affine coordinate system $O_{\vec{e_1}\vec{e_2}\vec{e_3}}$, defined by the accelerometers' axis, $\overline{a_x}$, $\overline{a_y}$, $\overline{a_z}$, respectively. Let the corresponding coordinates in the $O_{\vec{f_1}\vec{f_2}\vec{f_3}}$ coordinate system (defined as in Proposition \ref{prop:orthSyst}) be
$\overline{a_x}^{orth}$, $\overline{a_y}^{orth}$ and $\overline{a_z}^{orth}$. Then we have
\begin{equation}
\left[
	\begin{array}{c}
		\overline{a_x}^{orth}\\
		\overline{a_y}^{orth}\\
		\overline{a_z}^{orth}
	\end{array}
\right]=\overline{T}^{orth} \left[
	\begin{array}{c}
		\overline{a_x}\\
	\overline{a_y}\\
	\overline{a_z}
	\end{array}
\right],
\end{equation}
where $\overline{T}^{orth}$ is defined as (\ref{eq:125}).
\label{cor:1}
\end{corollary}
Let us denote 
\[\bm{\overline{a}}^{orth}:=\left[
	\begin{array}{c}
		\overline{a_x}^{orth}\\
		\overline{a_y}^{orth}\\
		\overline{a_z}^{orth}
	\end{array}
\right],\ \bm{\overline{a}}:=\left[
	\begin{array}{c}
		\overline{a_x}\\
	\overline{a_y}\\
	\overline{a_z}
	\end{array}
\right].\]
Furthermore, let is take into account (\ref{theModelMatrix}), Remark \ref{remark:1} and Corollary \ref{cor:1} and the notation, introduced in those. Then we obtain the following important result:
\begin{theorem}
Let $\widehat{a_x}$, $\widehat{a_y}$ and $\widehat{a_z}$ be the values shown by the accelerometers at a certain time and $\overline{a_x}^{orth}$, $\overline{a_y}^{orth}$ and $\overline{a_z}^{orth}$ be the coordinates of the corresponding acceleration in the orthonormal coordinate system, defined in Proposition \ref{prop:orthSyst}. Then the following holds true:
\[\bm{\overline{a}}^{orth}=\overline{T}^{orth}\overline{T}T\bm{\widehat{a}}-\overline{T}^{orth}\overline{T}\bm{s}.\]
\end{theorem}
\section{Estimates for the Error Due to Nonorthogonality}
We want to evaluate what is the error that can be caused by not 
considering the nonorthogonality of the axis. 
Let us denote
\[err(a_x,a_y,a_z,\phi,\psi,\theta):=|a_{nonorth}-a_{orth}|.\]
Since in the calibration procedure we use the fact that when the accelerometers are at rest, then
the only acceleration that they measure is due to gravity, we consider the following extremal problem:
\begin{problem}
\begin{equation}
\begin{aligned}
&err(a_x,a_y,a_z,\phi,\psi,\theta) \longrightarrow \max,\\
&a_{nonorth}=g,\\
&0.98\pi\le\phi,\psi,\theta\le1.02\pi,
\end{aligned}
\end{equation}
where $g$ is the gravitational acceleration.
\end{problem}
Let us note that the restrictions for $\phi$, $\psi$, $\theta$ are imposed because of the up to 2\% nonorthogonality
of the axis. 

By solving this problem numerically, we obtain the estimate of the maximal error
\[\max err(a_x,a_y,a_z,\phi,\psi,\theta)\approx0.3130299\ m/s^2\]
This value is reached, for example at 
\[a_x=a_y=a_z\approx5.48114, \phi=\psi=\theta\approx1.60221\]

Next, we want to give an estimate for the upper bound of the relative error that can be reached when
the accelerometer measures an arbitrary acceleration. We assume consecutively that in each 
direction the accelerometer can measure acceleration up to $\pm4g$, $\pm8g$, $\pm16g$. In all the cases
the estimate is approximately the same.
\begin{problem}
\begin{equation}
\begin{aligned}
&\left|\frac{a_{nonorth}-a_{orth}}{a_{nonorth}}\right| \longrightarrow \max,\\
&-16g\le a_x,a_y,a_z \le 16g,\\
&0.98\pi\le\phi,\psi,\theta\le1.02\pi,
\end{aligned}
\end{equation}
\end{problem}
The numerical solution of this extremal problem is approximately 3.192\%.

Next, we present a histogram which represents the distribution of the relative error for
values of $\phi$, $\psi$, $\theta$, respectively 1.53938, 1.60221, 1.60221, which are the
estimated with the calibration procedure values for the sensor used in the tests (see Fig.\ref{histogram03}).
\begin{figure}[h]
\centering
\includegraphics[width=0.75\textwidth]{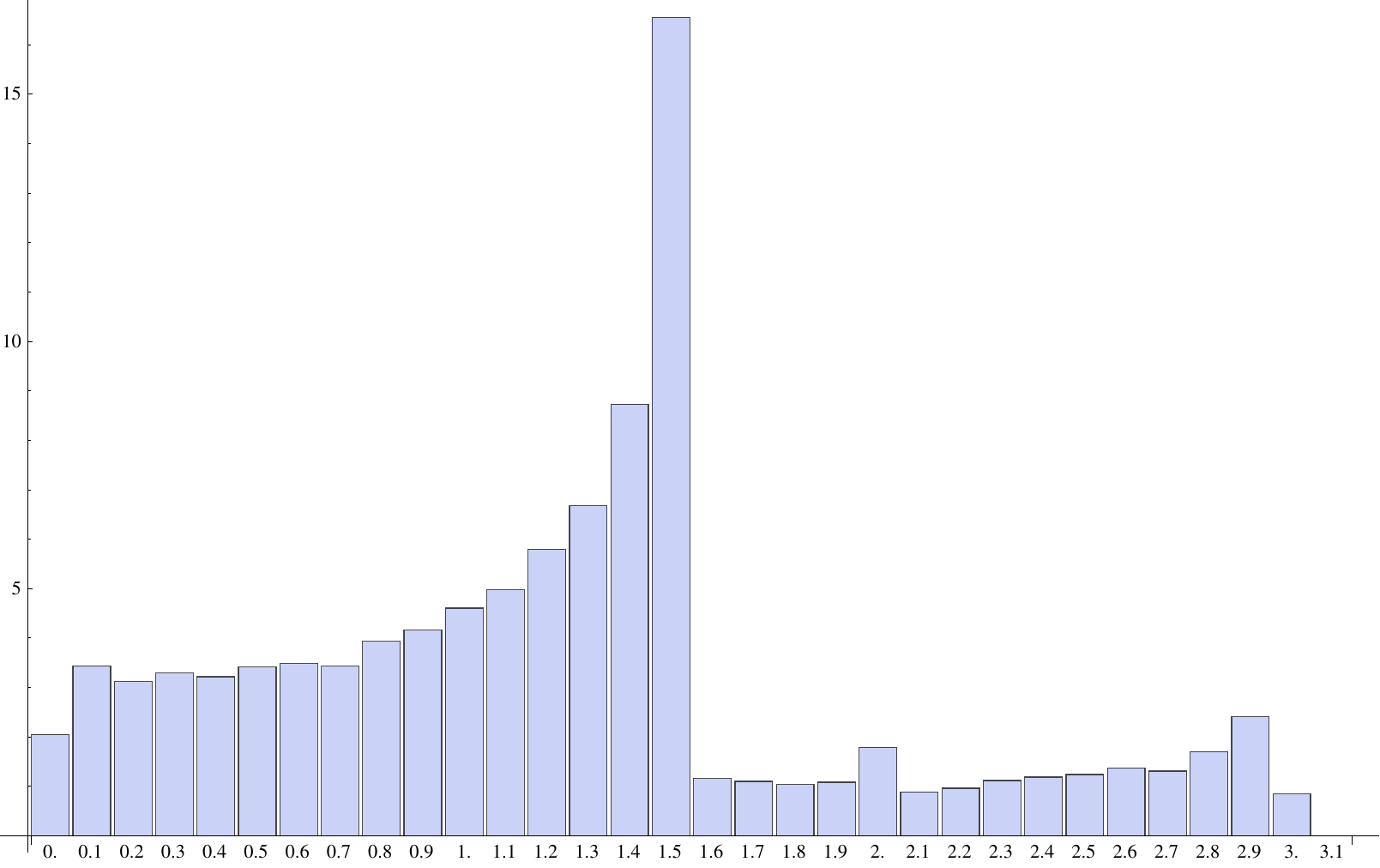}
\caption{Distribution of the relative error}
\label{histogram03}
\end{figure}

In order to be able to understand the ''jump'' in the histogram (Fig.\ref{histogram03}), we 
present geometrically the domains in $\mathbb{R}^3$ where the relative error is respectively $\le 1.4\%$, $\ge 1.4\%$,, $\le 1.6\%$,
$\ge 1.6\%$ (see Fig.\ref{relErrorFigs}). We assume again that $\phi$, $\psi$, $\theta$ are, respectively 1.53938, 1.60221, 1.60221, and $a_x$, $a_y$, $a_z$ lie in the interval $[-20,20]$.

\begin{figure}[h!]
\centering
\begin{subfigure}[b]{0.45\textwidth}
\includegraphics[width=\textwidth]{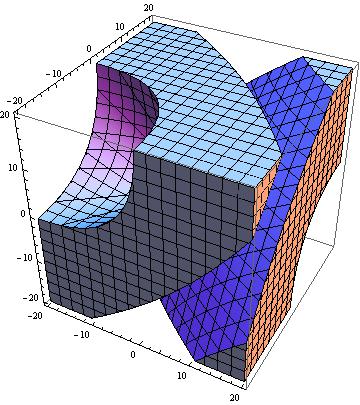}
\caption {relative error $\le 1.4\%$}
\end{subfigure}
\begin{subfigure}[b]{0.45\textwidth}
\includegraphics[width=\textwidth]{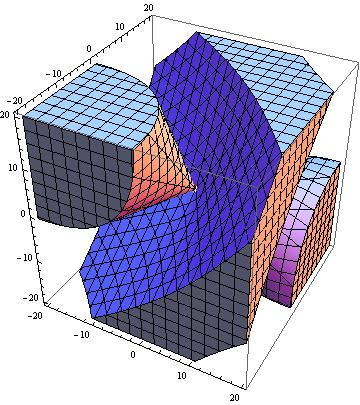}
\caption {relative error $\ge 1.4\%$}
\end{subfigure}

\begin{subfigure}[b]{0.45\textwidth}
\includegraphics[width=\textwidth]{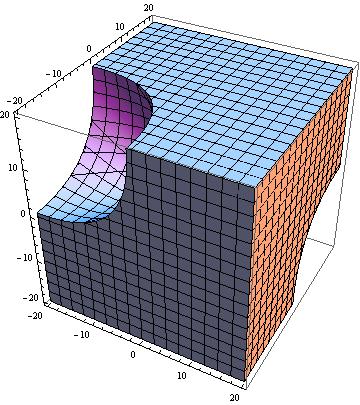}
\caption {relative error $\le 1.6\%$}
\end{subfigure}
\begin{subfigure}[b]{0.45\textwidth}
\includegraphics[width=\textwidth]{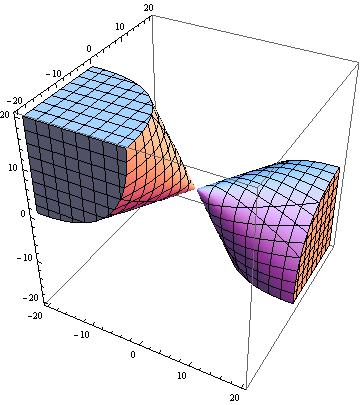}
\caption {relative error $\ge 1.6\%$}
\end{subfigure}
\caption{Domains in $\mathbb{R}^3$ for the relative error}
\label{relErrorFigs}
\end{figure}
\pagebreak

\section{Conclusion}
We have derived an explicit formula for the linear acceleration, given the values of the accelerometers and the angles between the axis of the sensors. We have explained how one can find an approximation of those angles and the shifts and scaling coefficients, given some measurements of the sensors, when the system does not undergo any acceleration. We have run some numerical experiments in order to understand better the error, caused by the nonorthogonality of the axis.

Future work on the topic could include studying the error with which the calibration parameters are evaluated with the procedure that we propose. Also, the optimal positions in which the calibration data is to be collected, should be determined. Probably, more investigation on the errors can give some insights about those positions.

\section*{Acknowledgment}
The authors would like to express their sincere gratitude to Assoc. Prof. Kiril Alexiev for his time and  valuable advices, that have helped the present work.

\end{document}